\documentclass{article}
\usepackage{amsmath,amssymb}

\def\B{{\cal B}}
\def\C{{\cal C}}
\def\D{{\cal D}}
\def\M{{\cal M}}
\def\N{{\cal N}}
\def\F{{\cal F}}

\def\reals{\mathbb{R}}
\def\ereals{\overline{\mathbb{R}}}
\def\uball{\mathbb{B}}

\def\comp{\raise 1pt \hbox{$\scriptstyle\circ$}}

\def\dom{\mathop{\rm dom}}

\def\upto{{\raise 1pt \hbox{$\scriptstyle \,\nearrow\,$}}}
\def\downto{{\raise 1pt \hbox{$\scriptstyle \,\searrow\,$}}}

\def\rc{\mathop{\rm rc}}
\def\epi{\mathop{\rm epi}}
\def\pos{\mathop{\rm pos}}

\def\FF{(\F_t)_{t=0}^T}

\newtheorem{theorem}{Theorem}
\newtheorem{proposition}[theorem]{Proposition}
\newtheorem{lemma}[theorem]{Lemma}
\newtheorem{corollary}[theorem]{Corollary}
\newtheorem{definition}[theorem]{Definition}
\newtheorem{example}[theorem]{Example}
\newtheorem{remark}[theorem]{Remark}
\newenvironment{proof}
{\begin{trivlist}\item[\, 
{\bf Proof.}]}{{\hfill $\square$}\end{trivlist}}

\title{Superhedging in illiquid markets}
\author{Teemu Pennanen\footnote{Department of Mathematics and Systems Analysis, Helsinki University of Technology, P.O.Box 1100, FI-02015 TKK, Finland, {\tt teemu.pennanen@tkk.fi}}}

\begin{document}
\maketitle

\begin{abstract}
We study contingent claims in a discrete-time market model where trading costs are given by convex functions and portfolios are constrained by convex sets. In addition to classical frictionless markets and markets with transaction costs or bid-ask spreads, our framework covers markets with nonlinear illiquidity effects for large instantaneous trades. We derive dual characterizations of superhedging conditions for contingent claim processes in a market without a cash account. The characterizations are given in terms of stochastic discount factors that correspond to martingale densities in a market with a cash account. The dual representations are valid under a topological condition and a weak consistency condition reminiscent of the ``law of one price'', both of which are implied by the no arbitrage condition in the case of classical perfectly liquid market models. We give alternative sufficient conditions that apply to market models with nonlinear cost functions and portfolio constraints.
\end{abstract}

\section{Introduction}

The notion of arbitrage is often given a central role when studying pricing and hedging of contingent claims in financial markets. In classical perfectly liquid market models, there are two good reasons for this. First, a violation of the no arbitrage condition leads to an unnatural situation where one can find self-financing trading strategies that generate infinite proceeds out of zero initial investment. Second, as discovered by Schachermayer~\cite{sch92}, the no arbitrage condition implies the closedness of the set of claims that can be superhedged with zero cost. The closedness yields dual characterizations of superhedging conditions in terms of e.g.\ martingale measures and state price deflators. 

In illiquid markets, however, things are different. A violation of the no arbitrage condition does not necessarily mean that one can generate infinite proceeds by simple scaling of arbitrage strategies. Indeed, illiquidity effects may come into play when trades get larger; see {\c{C}}etin and Rogers~\cite{cr7} or Pennanen~\cite{pen7,pen8}. On the other hand, even in the case of linear models, the no arbitrage condition is not necessary for closedness of the set of claims hedgeable with zero cost. There may exist other economically meaningful conditions that yield the closedness and corresponding dual characterizations of superhedging conditions.

This paper studies superhedging in a nonlinear discrete time model from \cite{pen7,pen8} where trading costs are given by convex cost functions and portfolios may be constrained by convex constraints. The model generalizes many better-known models such as the classical linear model, the transaction cost model of Jouini and Kallal~\cite{jk95a}, the sublinear model of Kaval and Molchanov~\cite{km5}, the illiquidity model of {\c{C}}etin and Rogers~\cite{cr7} as well as the linear models with portfolio constraints of Pham and Touzi~\cite{pt99}, Napp~\cite{nap3}, Evstigneev, Sch{\"u}rger and Taksar~\cite{est4} and Rokhlin~\cite{rok5b}. Our model covers nonlinear illiquidity effects associated with instantaneous trades (market orders) but it assumes that agents have no market power in the sense that trades do not affect the costs of subsequent trades. This is analogous to the models of {\c{C}}etin, Jarrow and Protter~\cite{cjp4}, {\c{C}}etin, Soner and Touzi~\cite{cst7} and Rogers and Singh~\cite{rs7}, the last one of which gives economic motivation for the assumption. We avoid long term price impacts because they interfere with convexity which is essential in many aspects of pricing and hedging. Convexity becomes an important issue also in numerical calculations; see e.g.\ Edirisinghe, Naik and Uppal~\cite{enu93}.

Unlike in most of the above papers, we do not assume the existence of a cash account a priori. This is important when studying {\em contingent claim processes} which may give pay-outs not only at maturity but throughout the whole life time of the claim. Such claim processes are common in practice where wealth cannot be transferred quite freely in time. Accordingly, much as in Jaschke and K\"uchler~\cite{jk1}, we study pricing in terms of {\em premium processes} instead of a single premium (price) in the beginning. This is quite natural since much of trading in practice consists of exchanging sequences of cash flows. The usual pricing and hedging problems are obtained as special cases when the premium process is null after the initial date and claims have pay-outs only at maturity. There is a simple and quite natural condition under which these more general pricing problems are well-defined and nontrivial in convex, possibly nonlinear market models. The condition is reminiscent of the ``law of one price'' (or the ``no good deal of the second kind'' in \cite{jk1}) which is weaker that the no arbitrage condition. 

We derive dual characterizations of superhedging in terms of deflators that correspond to martingale densities in a market with a cash account. In the presence of nonlinearities, a new term appears in pricing formulas much like in dual representations of convex risk measures which are not positively homogeneous. This was observed in F\"ollmer and Schied~\cite[Proposition~16]{fs2} in the presence of convex constraints in the classical linear model with a cash-account; see also Staum~\cite{sta4}, Frittelli and Scandolo~\cite{fs6}, and Kl\"oppel and Scweizer~\cite[Section~4]{ks7}.

Our duality results hold under the assumption that the set of claims hedgeable with zero cost is closed in probability. The closedness condition is known to be satisfied under the no arbitrage condition in classical perfectly liquid models \cite{sch92}. Other sufficient conditions for the closedness can be derived from the results of Schachermayer~\cite{sch4}, Kabanov, Rasonyi and Stricker~\cite{krs3} as well as the forthcoming paper Pennanen and Penner~\cite{pp8}. Whereas \cite{sch4} and \cite{krs3} deal with conical models, \cite{pp8} allows for more general convex models. In these papers, closedness is obtained for the larger set of portfolio-valued claims, which is not necessary when studying claims with cash delivery. More importantly, none of these papers allows for portfolio constraints. We give sufficient conditions that apply to claims with cash delivery under general nonlinear cost functions and portfolio constraints.

The rest of this paper is organized as follows, The next section defines the market model. Section~\ref{sec:sh} states the superhedging and pricing problems for claim and portfolio processes and studies their properties in algebraic terms. Section~\ref{sec:dual} derives dual characterizations of the superhedging conditions for integrable processes in terms of bounded deflators. This is done under the assumption that the set of claims hedgeable with zero cost is closed in probability. Section~\ref{sec:cl} derives sufficient conditions for the closedness. Section~\ref{sec:conclusions} makes some concluding remarks.

\section{The market model}\label{sec:mm}

Consider a financial market where trading occurs over finite discrete time $t=0,\ldots,T$. Let $(\Omega,\F,P)$ be a probability space with a filtration $\FF$ describing the information available to an investor at each $t=0,\ldots,T$. For simplicity, we assume that $\F_0$ is the trivial $\sigma$-algebra $\{\emptyset,\Omega\}$ and that each $\F_t$ is complete with respect to $P$. The Borel $\sigma$-algebra on $\reals^J$ will be denoted by $\B(\reals^J)$.

\begin{definition}\label{ccp}
A \em convex cost process \em is a sequence $S=(S_t)_{t=0}^T$ of extended real-valued functions on $\reals^J\times\Omega$ such that for $t=0,\ldots,T$,
\begin{enumerate}
\item
the function $S_t(\cdot,\omega)$ is convex, lower semicontinuous and vanishes at $0$ for every $\omega\in\Omega$,
\item
$S_t$ is $\B(\reals^J)\otimes\F_t$-measurable.
\end{enumerate}
A cost process $S$ is said to be {\em nondecreasing, nonlinear, polyhedral, positively homogeneous, linear,} \ldots if the functions $S_t(\cdot,\omega)$ have the corresponding property for every $\omega\in\Omega$.
\end{definition}

The interpretation is that buying a portfolio $x_t\in\reals^J$ at time $t$ and state $\omega$ costs $S_t(x_t,\omega)$ units of cash. The measurability property implies that if the portfolio $x_t$ is $\F_t$-measurable then the cost $\omega\mapsto S_t(x_t(\omega),\omega)$ is also $\F_t$-measurable (see e.g.\ \cite[Proposition~14.28]{rw98}). This just means that the cost is known at the time of purchase. We pose no smoothness assumptions on the functions $S_t(\cdot,\omega)$. The measurability property together with lower semicontinuity in Definition~\ref{ccp}  mean that $S_t$ is an $\F_t$-measurable \em normal integrand \em in the sense of Rockafellar~\cite{roc68}; see also Rockafellar and Wets~\cite[Chapter 14]{rw98}.

Definition~\ref{ccp}, originally given in \cite{pen6}, was motivated by the structure of double auction markets where the costs of market orders are polyhedral convex functions of the number of shares bought. The classical linear market model corresponds to $S_t(x,\omega)=s_t(\omega)\cdot x$, where $s=(s_t)_{t=0}^T$ is an $\reals^J$-valued $\FF$-adapted price process. Definition~\ref{ccp} covers also many other models from literature; see \cite{pen6}.

We allow for general convex portfolio constraints where at each $t=0,\ldots,T$ the portfolio $x_t$ is restricted to lie in a convex set $D_t$ which may depend on $\omega$.

\begin{definition}\label{pcp}
A {\em convex portfolio constraint process} is a sequence $D=(D_t)_{t=0}^T$ of set-valued mappings from $\Omega$ to $\reals^J$ such that for $t=0,\ldots,T$,
\begin{enumerate}
\item
$D_t(\omega)$ is closed, convex and $0\in D_t(\omega)$ for every $\omega\in\Omega$,
\item
the set-valued mapping $\omega\mapsto D_t(\omega)$ is $\F_t$-measurable.
\end{enumerate}
A constraint process $D$ is said to be {\em polyhedral, conical,} \ldots if the sets $D_t(\omega)$ have the corresponding property for every $\omega\in\Omega$.
\end{definition}

The classical case without constraints corresponds to $D_t(\omega)=\reals^J$ for every $\omega\in\Omega$ and $t=0,\ldots,T$. In addition to obvious ``short selling'' restrictions, portfolio constraints can be used to model situations where one encounters different interest rates for lending and borrowing. This can be done by introducing two separate ``cash accounts'' whose unit prices appreciate according to the two interest rates and restricting the investments in these assets to be nonnegative and nonpositive, respectively. A simple example that goes beyond conical and deterministic constraints is when there are nonzero bounds on market values of investments.

\begin{remark}[Market values]\label{rem:mv}
Large investors usually view investments in terms of their market values rather than in units of shares; see \cite{kab99} and \cite{mp8}. If $s=(s_t)_{t=0}^T$ is a componentwise strictly positive $\reals^J$-valued process, we can write
\[
S_t(x,\omega) = \varphi_t(M_t(\omega)x,\omega)
\]
where $\varphi_t(h):=S_t((h^j/s_t^j)_{j\in J})$ and $M_t(\omega)$ is the diagonal matrix with entries $s_t^j(\omega)$. Everything that is said below can be stated in terms of the variables $h_t^j(\omega) := s_t^j(\omega)x_t^j(\omega)$. If $s_t(\omega)$ is a vector of ``market prices'' of the assets $J$, then the vector $h_t(\omega)$ gives the ``market values'' of the assets held. Market prices are usually understood as the unit prices associated with infinitesimal trades. If the cost function $S_t(\cdot,\omega)$ is smooth at the origin, then $s_t(\omega)=\nabla S_t(0,\omega)$ is the natural definition. If $S_t(\omega)$ is nondifferentiable at the origin, then $s_t(\omega)$ could be any element of the {\em subdifferential}
\[
\partial S_t(0,\omega) :=\{s\in\reals^J\,|\,S_t(x,\omega)\ge S_t(0,\omega) + s\cdot x\quad\forall x\in\reals^J\}.
\]
In double auction markets, $\partial S_t(0,\omega)$ is the product of the intervals between the bid and ask prices of the assets $J$; see \cite{pen6}.
\end{remark}

\section{Superhedging}\label{sec:sh}

When wealth cannot be transfered freely in time (due to e.g.\ different interest rates for lending and borrowing) it is important to distinguish between payments/costs that occur at different dates. A {\em (contingent) claim process} is an $\reals$-valued stochastic process $c=(c_t)_{t=0}^{T}$ that is adapted to $(\F_t)_{t=0}^T$. The value of $c_t$ is interpreted as the amount of cash the owner of the claim receives at time $t$. Such claim processes are common e.g.\ in insurance. The set of claim processes will be denoted by $\M$. 

A {\em portfolio process}, is an $\reals^J$-valued stochastic process $x=(x_t)_{t=0}^T$ that is adapted to $(\F_t)_{t=0}^T$. The vector $x_t$ is interpreted as a portfolio that is held over the period $[t,t+1]$. The set of portfolio processes will be denoted by $\N$. An $x\in\N_0:=\{x\in\N\,|\,x_T=0\}$ superhedges a claim process $c\in\M$ with zero cost if it satisfies the {\em budget constraint}\footnote{Given an $\F_t$-measurable function $z_t:\Omega\to\reals^J$, $S_t(z_t)$ denotes the extended real-valued random variable $\omega\mapsto S_t(z_t(\omega),\omega)$. By \cite[Proposition~14.28]{rw98}, $S_t(z_t)$ is $\F_t$-measurable whenever $z_t$ is $\F_t$-measurable.} 
\[
S_t(x_t-x_{t-1}) + c_t \le 0\quad P\text{-a.s.}\quad t=0,\ldots,T.
\]
Here and in what follows, we always set $x_{-1}=0$. At the terminal date, we require that everything is liquidated so the budget constraint becomes $S_T(-x_{T-1})+c_T\le 0$. The above is a numeraire-free way of writing the superhedging property; see Example~\ref{ex:numeraire}. In the case of a stock exchange, the interpretation is that the portfolio is updated by market orders in a way that allows for delivering the claim without any investments over time. In particular, when $c_t$ is strictly positive, the cost $S_t(x_t-x_{t-1})$ of updating the portfolio from $x_{t-1}$ to $x_t$ has to be strictly negative (market order of $x_t-x_{t-1}$ involves more selling than buying).  

The set of all claim processes that can be superhedged with zero cost under constraints $D$ will be denoted by $\C$. That is, 
\[
\C = \{c\in\M\,|\,\exists x\in\N_0:\ x_t\in D_t,\ S_t(\Delta x_t)+c_t\le 0,\ t=0,\ldots,T\},
\]
where $\N_0=\{x\in\N\,|\,x_T=0\}$. Many aspects of superhedging can be conveniently studied in terms of this set. 

\begin{example}[Numeraire and stochastic integrals]\label{ex:numeraire}
Assume that there is a perfectly liquid asset, say $0\in J$, such that
\begin{align*}
S_t(x,\omega) &= x^0 + \tilde S_t(\tilde x,\omega),\\
D_t(\omega) &= \reals\times\tilde D_t(\omega),
\end{align*}
where $x=(x^0,\tilde x)$ and $\tilde S$ and $\tilde D$ are the cost process and the constraints for the remaining ``risky'' assets $\tilde J=J\setminus\{0\}$.
Given $\tilde x=(\tilde x_t)_{t=0}^T$, we can define 
\[
x_t^0 = x^0_{t-1} - \tilde S_t(\tilde x_t-\tilde x_{t-1}) -  c_t \quad t=0,\ldots,T-1,
\]
so that the budget constraint holds as an equality for $t=1,\ldots,T-1$ and
\[
x^0_{T-1} = - \sum_{t=0}^{T-1}\tilde S_t(\tilde x_t-\tilde x_{t-1}) - \sum_{t=0}^{T-1} c_t.
\]
We then get the expression
\[
\C = \{c\in\M\,|\,\exists\tilde x:\ \tilde x_t\in\tilde D_t,\ \sum_{t=0}^T\tilde S_t(\tilde x_t-\tilde x_{t-1}) + \sum_{t=0}^T c_t \le 0\}.
\]
If moreover, the cost process $\tilde S$ is linear, i.e.\ $\tilde S_t(\tilde x)=\tilde s_t\cdot\tilde x$ we have
\[
\sum_{t=0}^T\tilde S_t(\tilde x_t-\tilde x_{t-1}) = \sum_{t=0}^T\tilde s_t\cdot(\tilde x_t-\tilde x_{t-1}) = -\sum_{t=0}^{T-1}\tilde x_t\cdot(\tilde s_{t+1}-\tilde s_t),
\]
so that
\[
\C = \{c\in\M\,|\,\exists\tilde x:\ \tilde x_t\in\tilde D_t,\ \sum_{t=0}^T c_t \le \sum_{t=0}^{T-1}\tilde x_t\cdot(\tilde s_{t+1}-\tilde s_t)\}.
\]
Thus, when a numeraire exists, hedging of a claim process can be reduced to hedging cumulated claims at the terminal date and if the cost process is linear, the hedging condition can be written in terms of a stochastic integral. This is essentially the market model studied e.g.\ in \cite{fs2}, \cite[Chapter~9]{fs4} and \cite[Section~4]{ks7}, where constraints on the risky assets were considered.
\end{example}

In problems of superhedging, one usually looks for the initial endowments (premiums) that allow, without subsequent investments, for delivering a claim with given maturity. Since in illiquid markets, cash at different dates are genuinely different things, it makes sense to study superhedging in terms of ``premium processes''. A {\em premium process} is a real-valued adapted stochastic process $p=(p_t)_{t=0}^T$ of cash flows that the seller receives in exchange for delivering a claim $c=(c_t)_{t=0}^T$. We say that $p\in\M$ is a {\em superhedging premium} for $c\in\M$ if there exists a portfolio process $x\in\N_0$ such that
\[
x_t\in D_t,\quad S_t(x_t-x_{t-1}) + c_t \le p_t
\]
almost surely for every $t=0,\ldots,T$. This can be written as 
\[
c-p\in \C.
\]
We are thus looking at situations where one sequence of payments is exchanged for another and the problem is to characterize those exchanges where residual risks can be completely hedged by an appropriate trading strategy.

Much research has been devoted to the case where premium is paid only in the beginning and claims only at the end. This corresponds to the case $p=(p_0,0,\ldots,0)$ and $c=(0,\ldots,0,c_T)$. The dynamic framework above is not only mathematically convenient (claims and premiums live in the same space) but also practical since much of trading consists of exchanging sequences of cash flows. This is the case e.g.\ in swap contracts where a stochastic sequence of payments is exchanged for a sequence of deterministic ones. Also, in various insurance contracts premiums are paid annually instead of a single payment in the beginning.

In some situations, a premium process $p\in\M$ is given and the question is what multiple of $p$ will be sufficient to hedge a claim $c\in\M$. This is the case e.g.\ in some defined benefit pension plans where the premium process is a fraction (the contribution rate) of the salary of the insured. In swap contracts, the premium process is often defined as a multiple of a constant sequence. Given a premium process $p\in\M$, we define the function $\pi:\M\to\ereals$ by
\[
\pi(c) := \inf\{\alpha\,|\,c-\alpha p\in \C\}.
\]
In the case $p=(1,0,\ldots,0)$, $\pi(c)$ gives the least initial investment that allows for delivering $c\in\M$ without risk or subsequent investments. The above definition of $\pi$ with a general $p\in\M$ is similar to \cite{jk1} where risk measures on general vector spaces were studied. We call $\pi(c)$ the {\em super hedging cost} of $c$.

Our subsequent analysis will be largely based on the following simple observation from \cite{pen6}. Here $\M_-$ denotes the set of nonpositive claim processes and 
\[
\rc C := \{c\,|\,c'+\alpha c\in C\quad\forall c'\in C,\ \alpha>0\}
\]
denotes the recession cone of a nonempty set $C$; see~\cite[Section~8]{roc70a}. Note that if $C$ is a cone, then $\rc C= C$.

\begin{lemma}\label{lem:convex}
The set $\C$ is convex and $\M_-\subset\rc \C$. If $S$ is sublinear and $D$ is conical, then $\C$ is a cone.
\end{lemma}


It is natural to assume that the premium $p\in\M$ is desirable in the sense that for any $c\in \C$,
\[
c-\alpha p\in \C\quad\forall\alpha>0,
\]
i.e.\ that $-p\in\rc \C$. This condition was proposed already in \cite[page 187]{jk1} in the context of more abstract acceptance sets. Since $\M_-\subset\rc \C$, the condition $-p\in\rc \C$ holds in particular if $p$ is in the set $\M_+$ of nonnegative claims. This certainly holds with the traditional choice $p=(1,0,\ldots,0)$. On the other hand, we should have $p\notin\rc \C$ since otherwise, $\pi(c)=-\infty$ for every 
\[
c\in\dom\pi:=\{c\in\M\,|\,\pi(c)<+\infty\}.
\]
In a market with a cash account (see Example~\ref{ex:numeraire}), the condition $(1,0,\ldots,0)\notin\rc \C$ means that every claim leaves the set $\C$ if a sufficiently large positive constant is added to it.

The existence of a $p\in\M$ such that $-p\in\rc \C$ and $p\notin\rc \C$ is equivalent to $\C\ne\M$ which just means that there is at least one claim that cannot be superhedged with zero cost. Indeed, there is no such $p$ iff $\rc \C$ is a linear space. Since $\M_-\subset\rc \C$, by Lemma~\ref{lem:convex}, this would mean that $\rc \C=\M$ which is equivalent to $\C=\M$.

\begin{proposition}\label{prop:pi}
Assume that $-p\in\rc \C$ and $p\notin\rc \C$. Then
\begin{enumerate}
\item
$\pi(\alpha_1c_1+\alpha_2c_2)\le\alpha_1 \pi(c_1)+\alpha_2 \pi(c_2)$ if $\alpha_i>0$ and $\alpha_1+\alpha_2=1$,
\item
$\pi(c_1)\le \pi(c_2)$ if $c_1\le c_2$,
\item
$\pi(c+\alpha p)=\pi(c)+\alpha$ for $\alpha\in\reals$ and $c\in\M$,
\item
$\pi(0)\le 0$.
\end{enumerate}
If $\C$ is a cone, then
\begin{enumerate}
\item[5.]
$\pi(\alpha c)=\alpha \pi(c)$ for $\alpha >0$.
\end{enumerate}
If $\C$ is algebraically closed, then 
\begin{enumerate}
\item[6.]
$\pi$ is proper, i.e.\ $\pi(c)>-\infty$ for all $c\in\M$,
\item[7.]
$\epi\pi:=\{(c,\alpha)\,|\,\pi(c)\le\alpha\}=\{(c,\alpha)\,|\,c-\alpha p\in \C\}$.
\end{enumerate}
\end{proposition}

\begin{proof}
Let $c_i\in\dom\pi$, $\alpha_i>\pi(c_i)$ and $\lambda\in(0,1)$ be arbitrary. Since $-p\in\rc \C$, we have $c_i-\alpha_ip\in \C$ and since $\C$ is convex, 
\[
\lambda c_1+(1-\lambda)c_2 -(\lambda\alpha_1+(1-\lambda)\alpha_2)p = \lambda(c_1-\alpha_1p)+(1-\lambda)(c_2-\alpha_2p)\in \C.
\]
Thus, $\pi(\lambda c_1+(1-\lambda)c_2)\le \lambda\alpha_1+(1-\lambda)\alpha_2$ and the convexity of $\pi$ follows. The monotonicity property follows from $\M_-\subset\rc \C$. The translation property is immediate from the definition of $\pi$ and $\pi(0)\le 0$ holds because $0\in \C$. When $\C$ is algebraically closed, the infimum $\pi(c)=\inf\{\alpha\,|\, c-\alpha p\in \C\}$ is finite and attained for every $c\in\dom\pi$, by Lemma~\ref{lem:rc}, since $p\notin\rc \C$. This gives 6 and 7.
\end{proof}

Proposition~\ref{prop:pi} shows that the super hedging cost $\pi$ has properties close to those of a convex risk measure; see e.g.\ \cite[Chapter~4]{fs4}. As a result, many of the existing results for risk measures can be applied to $\pi$; see Section~\ref{sec:dual}. We emphasize that we do not insist on having $\pi(0)\ge 0$. The condition $\pi(0)\ge 0$ would mean that $-\alpha p\notin \C$ for all $\alpha<0$ or equivalently that $p\notin \pos \C$, where
\[
\pos \C := \bigcup_{\alpha >0}\alpha \C
\]
is a convex cone known as the {\em positive hull} of $\C$. When $\C$ is a cone, we have $\pos \C = \rc \C = \C$ and then $\pi(0)\ge 0$ as soon as $p\notin\rc \C$. In general, however, the condition $p\notin\rc \C$ is weaker than $\pi(0)\ge 0$. The condition $\pi(0)\ge 0$ is related to the well known ``law of one price''. In the case $p=(1,0,\ldots,0)$ the two conditions would be equivalent if one required exact replication instead of superhedging in the definition of $\pi$; see \cite{cdks4}. The condition $p\notin\pos\C$ is essentially what would be called ``no good deal of the second kind'' in the terminology of \cite[page 193]{jk1}.

The nonpositive number $\pi(0)$ is the smallest multiple of the premium $p$ one needs in order to find a riskless strategy in the market. If one has to deliver a claim $c\in\M$, one needs $\pi(c)-\pi(0)$ euros more. More generally, we define the {\em superhedging selling price} of a $c\in\M$ for an agent with {\em initial liabilities} $\bar c\in\M$ as 
\[
P(\bar c;c) = \pi(\bar c+c)-\pi(\bar c).
\]
Analogously, the {\em superhedging buying price} of a $c\in\M$ for an agent with {\em initial liabilities} $\bar c\in\M$ is given by $\pi(\bar c)-\pi(\bar c-c)=-P(\bar c;-c)$. It follows from convexity of $\pi$ that 
\[
P(\bar c;c)\ge -P(\bar c;-c),
\]
which means that agents with similar liabilities and similar market expectations should not trade with each other if they aim at superhedging their positions.

It is intuitively clear that the value an agent assigns to a claim should depend not only on the market expectations but also on the liabilities the agent might have already before the trade. For example, the selling price $P(\bar c^{IC};c)$ of a home insurance contract $c\in\M$ for an insurance company may be lower than the buying price $-P(\bar c^{HO};-c)$ for a home owner, even if the two had identical market expectations. Here $\bar c^{IC}$ would be the claims associated with the existing insurance portfolio of the company while $\bar c^{HO}$ would be the possible losses to the home owner associated with damages to the home. Another example would be the exchange of futures contracts between a wheat farmer and a wheat miller. In fact, many derivative contracts exist precisely because of the differences between initial liabilities of different parties. 


\begin{remark}
One could define the {\em  marginal selling price} of a claim $c\in\M$ given initial liabilities $\bar c\in\M$ as the directional derivative
\[
\pi'(\bar c;c):=\lim_{\alpha\downto 0}\frac{\pi(\bar c+\alpha c)-\pi(\bar c)}{\alpha} = \lim_{\alpha\downto 0}\frac{P(\bar c;\alpha c)}{\alpha}.
\]
Since $\pi$ is convex the limit is a well-defined and equals the infimum over $\alpha>0$; see \cite[Theorem~23.1]{roc70a}. Moreover,
\[
-P(\bar c;-c) \le -\pi'(\bar c;-c) \le \pi'(\bar c;c) \le P(\bar c;c)
\]
for every $c,\bar c\in\M$. Equality holds in the middle when $p$ is differentiable at $\bar c$ in direction $c$. Closely related ideas have been employed e.g.\ in Davis~\cite{dav97} and \cite[Section~5]{sta4}.
\end{remark}

\section{Duality}\label{sec:proofs}\label{sec:dual}

We derive dual characterizations of superhedging using functional analytic techniques much as e.g.\ in \cite{sch92}, \cite{fs4} of \cite{ds6}. Due to possible nonlinearities, our model requires a bit more convex analysis than the traditional linear models. In particular, a major role is played by the theory of {\em normal integrands} (see e.g.~\cite{roc68,roc76,rw98}), which explains the precise form of Definitions~\ref{ccp} and \ref{pcp}.

We first give dual characterizations of superhedging conditions in terms of the ``support function'' of the set of integrable claims in $\C$ under the assumption that $\C$ is closed in probability. We then give an expression for the support function in terms of $S$ and $D$, which allows for a more concrete characterizations of superhedging. Sufficient conditions for the closedness of $\C$ will be given in Section~\ref{sec:cl}.

Let $\M^1$ and $\M^\infty$ be the spaces of integrable and essentially bounded, respectively, real-valued adapted processes. Let 
\[
\C^1:=\C\cap\M^1,
\]
be the set of integrable claim processes that can be superhedged with zero cost. The bilinear form
\[
(c,y)\mapsto E\sum_{t=0}^Tc_ty_t
\]
puts $\M^1$ and $\M^\infty$ in separating duality; see \cite{roc74}. One can then use classical convex duality arguments to describe hedging conditions. This will involve the {\em support function} $\sigma_{\C^1}:\M^\infty\to\ereals$ of $\C^1$ defined by
\[
\sigma_{\C^1}(y) = \sup_{c\in\C^1}E\sum_{t=0}^Tc_t y_t.
\]
It is a nonnegative extended real-valued sublinear function on $\M^\infty$. Since $\C^1$ contains all nonpositive claim processes, the effective domain 
\[
\dom\sigma_{\C^1} = \{y\in\M^\infty\,|\,\sigma_{\C^1}(y)<\infty\}
\]
of $\sigma_{\C^1}$ is contained in the nonnegative cone $\M^\infty_+$. Moreover, since $\sigma_{\C^1}$ is sublinear, $\dom\sigma_{\C^1}$ is a convex cone. In the terminology of microeconomic theory, $\sigma_{\C^1}$ is called the {\em profit function} associated with the ``production set'' $\C^1$; see e.g.\ Aubin~\cite{aub79} or Mas-Collel, Whinston and Green~\cite{mwg95}. 

\begin{theorem}\label{thm:d}
Assume that $\C$ is closed in probability. Then
\begin{enumerate}
\item
$c\in \C^1$ if and only if $E\sum_{t=0}^Tc_ty_t \le 1$ for every $y\in\M^\infty$ such that $\sigma_{\C^1}(y)\le 1$,
\item
If $p\in\M^1$ is such that $-p\in\rc \C$ and $p\notin\rc \C$ then $\pi$ is a proper lower semicontinuous (both in norm and the weak topology) convex function on $\M^1$ with the representation
\[
\pi(c) =\sup_{y\in\M^\infty}\left\{\left.E\sum_{t=0}^Tc_ty_t - \sigma_{\C^1}(y)\,\right|\, E\sum_{t=0}^Tp_ty_t=1\right\}.
\]
\end{enumerate}
\end{theorem}

\begin{proof}
If $\C$ is closed in probability then $\C^1$ is be closed in the norm topology of $\M^1$ and the first claim is nothing but the classical bipolar theorem (see e.g.\ \cite[Theorem~14.5]{roc70a} or \cite[Section~1.4.2]{aub79}).

When $p\in\M^1$, the restriction $\bar\pi$ of $\pi$ to $\M^1$ can be written as $\bar\pi(c)=\inf\{\alpha\,|\,c-\alpha p\in \C^1\}$. The convex conjugate $\bar\pi^*:\M^\infty\to\ereals$ of $\bar\pi$ can be expressed as
\begin{align*}
\bar\pi^*(y) &= \sup_{c\in\M^1}\{E\sum_{t=0}^Tc_ty_t-\pi(c)\}\\
&= \sup_{c\in\M^1,\alpha\in\reals}\{E\sum_{t=0}^Tc_ty_t-\alpha\,|\,c-\alpha p\in \C^1\}\\
&= \sup_{c'\in\M^1,\alpha\in\reals}\{E\sum_{t=0}^T(c_t'+\alpha p_t)y_t-\alpha\,|\,c'\in \C^1\}\\
&= \sup_{c'\in\M^1,\alpha\in\reals}\{E\sum_{t=0}^Tc'_ty_t+\left(E\sum_{t=0}^Tp_ty_t-1\right)\alpha\,|\,c'\in \C^1\}\\
&=
\begin{cases}
\sigma_{\C^1}(y) & \text{if $E\sum_{t=0}^Tp_ty_t=1$},\\
+\infty & \text{otherwise}.
\end{cases}
\end{align*}
The representation for $\pi$ on $\M^1$ thus means that $\bar\pi$ equals the conjugate of $\bar\pi^*$. By \cite[Theorem~5]{roc74}, this holds when $\bar\pi$ is lower semicontinuous and $\bar\pi(c)>-\infty$ for every $c\in\M^1$. Since, by assumption, $\C$ is closed in probability it is also algebraically closed and then, by Proposition~\ref{prop:pi}, $\bar\pi(c)>-\infty$ for every $c\in\M^1$ and $\epi\bar\pi=\{(c,\alpha)\in\M^1\times\reals\,|\,c-\alpha p\in \C^1\}$. Since $\C$ is closed in probability, we have that $\C^1$ is closed in norm and then $\epi\bar\pi$ is closed in the product topology of $\M^1\times\reals$, i.e.\ $\bar\pi$ is lower semicontinuous in norm. By the classical separation argument, a convex function which is lower semicontinuous in the norm topology is lower semicontinuous also in the weak topology.
\end{proof}

The function $\sigma_{\C^1}$ plays a similar role in superhedging of claim processes as the ``penalty function'' does in the theory of convex risk measures; see e.g.~\cite[Chapter~4]{fs4}. In the conical case, Theorem~\ref{thm:d} simplifies much like the dual representation of a coherent risk measure.

\begin{corollary}\label{cor:d}
Assume that $\C$ is conical and closed in probability. Denoting the polar cone of $\C^1$ by
\[
\D^\infty = \{y\in\M^\infty\,|\,E\sum_{t=0}^Tc_ty_t\le 0\ \forall c\in\C^1\},
\]
we have
\begin{enumerate}
\item
$c\in \C^1$ if and only if $E\sum_{t=0}^Tc_ty_t \le0$ for every $y\in\D^\infty$.
\item
If $p\in\M^1$ is such that $-p\in\C$ and $p\notin\C$ then $\pi$ is a proper lower semicontinuous (both in norm and the weak topology) sublinear function on $\M^1$ with the representation
\[
\pi(c) =\sup_{y\in\D^\infty}\left\{\left.E\sum_{t=0}^Tc_ty_t\,\right|\, E\sum_{t=0}^Tp_ty_t=1\right\}.
\]
\end{enumerate}
\end{corollary}

\begin{proof}
When $\C$ is a cone the set $\C^1$ is also a cone so that
\[
\sigma_{\C^1}(y)=
\begin{cases}
0 & \text{if $E\sum_{t=0}^Tc_ty_t\le 0$ for every $c\in\C^1$},\\
+\infty & \text{otherwise}.
\end{cases}
\]
Then, $\sigma_{\C^1}(y)\le 1$ iff $\sigma_{\C^1}(y)\le 0$ iff $y\in\D^\infty$.
\end{proof}

The superhedging condition for a premium process $p\in\M^1$ and a claim process $c\in\M^1$ can be written as $c-p\in \C$. If the premium is of the form $p=(p_0,0,\ldots,0)$ (single payment in the beginning) and the claim is of the form $c=(0,\ldots,0,c_T)$ (single payment at the end), then the first part of Corollary~\ref{cor:d} says that $p$ is a superhedging premium for $c$ if and only if 
\[
E(c_Ty_T)\le p_0y_0\quad\forall y\in\D^\infty.
\]
When $p=(1,0,\ldots,0)$, the equation $E\sum_{t=0}^Tp_ty_t=1$  in the representation for $\pi$ (both in Corollary~\ref{cor:d} and Theorem~\ref{thm:d}) becomes $y_0=1$.

When $S$ is integrable (see below), we can express $\sigma_{\C^1}$ and thus Theorem~\ref{thm:d} and Corollary~\ref{cor:d} more concretely in terms of $S$ and $D$. This will involve the space $\N^1$ of $\reals^J$-valued adapted integrable processes $v=(v_t)_{t=0}^T$ and the integral functionals 
\[
v_t\mapsto E(y_tS_t)^*(v_t)\quad\text{and}\quad v_t\mapsto E\sigma_{D_t}(v_t)
\]
associated with the normal integrands
\begin{align*}
(y_tS_t)^*(v,\omega) &:= \sup_{x\in\reals^J}\{x\cdot v - y_t(\omega)S_t(x,\omega)\}
\intertext{and}
\sigma_{D_t(\omega)}(v) &:= \sup_{x\in\reals^J}\{x\cdot v\,|\, x\in D_t(\omega)\}.
\end{align*}
That the above expressions do define normal integrands follows from \cite[Theorem~14.50]{rw98}. Since $S_t(0,\omega)=0$ and $0\in D_t(\omega)$ for every $t$ and $\omega$, the functions $(y_tS_t)^*$ and $\sigma_{D_t}$ are nonnegative. 

We will say that a cost process $S=(S_t)_{t=0}^T$ is {\em integrable} if the functions $S_t(x,\cdot)$ are integrable for every $t=0,\ldots,T$ and $x\in\reals^J$. In the classical linear case $S_t(x,\omega)=s_t(\omega)\cdot x$, integrability means that the marginal price $s$ is integrable in the usual sense. The following is from \cite{pen8}.

\begin{lemma}\label{lem:pid}
If $S$ is integrable, then for $y\in\M^\infty_+$, 
\[
\sigma_{\C^1}(y) = \inf_{v\in\N^1}\left\{\sum_{t=0}^T E(y_t S_t)^*(v_t) + \sum_{t=0}^{T-1}E\sigma_{D_t}(E[\Delta v_{t+1}|\F_t])\right\},
\]
while $\sigma_{\C^1}(y)=+\infty$ for $y\notin\M^\infty_+$. The infimum is attained for every $y\in\M^\infty_+$.
\end{lemma}

The expression for $\sigma_{\C^1}$ in Lemma~\ref{lem:pid} can be inserted in Theorem~\ref{thm:d} and Corollary~\ref{cor:d}. In particular, we get the following.



\begin{corollary}\label{cor:polar}
If $S$ is sublinear and integrable and $D$ is conical, then the polar of $\C^1$ can be expressed as
\[
\D^\infty=\{y\in\M^\infty_+\,|\, \exists s\in\N:\ s_t\in Z_t,\ E[\Delta(y_ts_t)\,|\,\F_{t-1}]\in D_t^*\quad t=1,\ldots T\},
\]
where $Z_t(\omega)=\{s\in\reals^J\,|\,s\cdot x\le S_t(x,\omega)\ \forall x\in\reals^J\}$, $D_t^*(\omega)$ is the polar cone of $D_t(\omega)$ and the condition $E[\Delta(y_ts_t)\,|\,\F_t]\in D_t^*$ includes the assumption that $\Delta(y_ts_t)$ is integrable.
\end{corollary}

\begin{proof}
If $S$ is sublinear and $D$ is conical, we have, by Theorems~13.1 and 13.2 of \cite{roc70a}, that
\[
(y_tS_t)^*(v,\omega) = 
\begin{cases}
0 & \text{if $v\in y_t(\omega)Z_t(\omega)$},\\
+\infty & \text{otherwise}
\end{cases}
\]
and
\[
\sigma_{D_t(\omega)}(v,\omega) = 
\begin{cases}
0 & \text{if $v\in D_t(\omega)^*$},\\
+\infty & \text{otherwise}.
\end{cases}
\]
By Lemma~\ref{lem:pid}, the polar cone $\D^\infty=\{y\in\M^\infty\,|\, \sigma_{\C^1}(y)\le 0\}$ can thus be written
\[
\D^\infty=\{y\in\M^\infty_+\,|\, \exists v\in\N^1:\ v_t\in y_tZ_t,\ E[\Delta v_t\,|\,\F_{t-1}]\in D_t^*\quad t=1,\ldots T\},
\]
so it suffices to make the substitution $v_t=y_ts_t$.
\end{proof}

When $S$ is linear with $S_t(x,\omega)=s_t(\omega)\cdot x$, we have $Z_t(\omega)=\{s_t(\omega)\}$, and when there are no portfolio constraints, i.e.\ $D_t(\omega)=\reals^J$, we have $D_t^*(\omega)=\{0\}$. Thus, in the linear unconstrained case,
Corollary~\ref{cor:polar} says that 
\[
\D^\infty=\{y\in\M^\infty_+\,|\, \text{$(y_ts_t)$ is a martingale}\}.
\]
When one of the assets has nonzero constant price (a numeraire exists), the martingale property means that $y_t=E[y_T\,|\,\F_t]$, so $y_T$ is a martingale density if $y_0=1$. When $p=(1,0,\ldots,0)$, the existence of such a $y\in\D^\infty$ is implied by the properness of $\pi$ in Corollary~\ref{cor:d}. We thus have that, in the linear unconstrained case, the existence of a martingale density follows from the closedness of $\C$ and the conditions $-(1,0,\ldots,0)\in\C$ and $(1,0,\ldots,0)\notin\C$ (since now $\rc\C=\C$, by Lemma~\ref{lem:convex}). The condition $-(1,0,\ldots,0)\in\C$ is immediate from the definition of $\C$, while $(1,0,\ldots,0)\notin\C$ means that it is not possible to be fully hedged when starting with strictly negative initial wealth. This last condition is much weaker than e.g.\ the no arbitrage condition. It is like the ``law of one price'' except that we allow the throwing away of money. On the other hand, the law of one price only yields the existence of a martingale which could take negative values; see \cite{cdks4}. 

While Theorem~\ref{thm:d} gives the existence of nontrivial deflators that yield the dual representation, the following characterizes the ones that attain the supremum in the representation.

\begin{proposition}
Assume that $\C$ is closed in probability and that $p\in\M^1$ is such that $-p\in\rc\C$ and $p\notin\rc\C$. Then the supremum in the dual representation of $\pi$ is attained by those $y\in\M^\infty$ which are dominated by the superhedging selling price with initial liabilities $c$, i.e.
\[
P(c;c')\ge E\sum_{t=0}^Tc'_ty_t\quad \forall c'\in\M^1.
\]
\end{proposition}

\begin{proof}
The dual representation can be written as
\[
\pi(c)=\sup_{y\in\M^\infty}\{E\sum_{t=0}^Tc_ty_t-\bar\pi^*(y)\},
\]
where $\bar\pi^*$ is the conjugate of the restriction of $\pi$ to $\M^1$; see the proof of Theorem~\ref{thm:d}. The supremum is attained by a $y\in\M^\infty$ iff 
\[
\pi(c)+\bar\pi^*(y)=E\sum_{t=0}^Tc_ty_t,
\]
which, by definition of $\bar\pi^*$, means that 
\[
\pi(c) + E\sum_{t=0}^T\tilde c_ty_t-\pi(\tilde c) \le E\sum_{t=0}^Tc_ty_t \quad\forall\tilde c\in\M^1.
\]
Setting $\tilde c=c+c'$, we can write this as
\[
\pi(c+c') -\pi(c) \ge E\sum_{t=0}^Tc'_ty_t \quad\forall c'\in\M^1.
\]
\end{proof}

\begin{remark}
The attainment of the supremum in the dual representation for $\pi(c)$ would be guaranteed by continuity of $\pi$ at $c$; see e.g.\ \cite[Section~6]{roc74}. In this case, the set of the attaining deflators would be nonempty and weak*-compact. In the sublinear case (when $\C$ is a cone), the converse holds, i.e.\ the weak*-compactness of the set of attaining deflators implies the continuity of $\pi$. We thus have a version of \cite[Corollary~4.35]{fs4} for the superhedging cost $\pi$. However, continuity of $\pi$ in $\M^1$ seems unlikely, unless $\Omega$ is finite. Consider, for example, a one period linear model with a cash account as e.g.\ in \cite[Chapter~1]{fs4}. If the price process $s$ is bounded, then $\dom\pi=\M^\infty$ which has empty interior in $\M^1$ so $\pi$ is nowhere continuous in $\M^1$.
\end{remark}

\section{Closedness of $\C$}\label{sec:cl}

In light of the above results, the closedness of $\C$ becomes an interesting issue. It was shown by Schachermayer~\cite{sch92} that when $S$ is a linear cost process with a cash account (see Example~\ref{ex:numeraire}) and $D=\reals^J$, the closedness is implied by the classical no arbitrage condition. In this section, we give sufficient conditions for other choices of $S$ and $D$.

In the classical linear model the finiteness of $\Omega$ is known to be sufficient for closedness even when there is arbitrage. More generally, we have the following. 

\begin{example}
If $S$ and $D$ are polyhedral and $\Omega$ is finite then $\C$ is closed. 
\end{example}

\begin{proof}
By \cite[Theorem~19.1]{roc70a} it suffices to show that $\C$ is polyhedral. The set $\C$ is the projection of the convex set
\[
E = \{(x,c)\in\N_0\times\M\,|\, x_t\in D_t,\ S_t(\Delta x_t)+c_t\le 0,\ t=0,\ldots,T\}.
\]
When $S$ and $D$ are polyhedral, we can describe the pointwise conditions $x_t\in D_t$ and $S_t(\Delta x_t)+c_t\le 0$ by a finite collection of linear inequalities. When $\Omega$ is finite, the set $E$ becomes an intersection of a finite collection of closed half-spaces. The set $\C$ is then polyhedral since it is a projection of a polyhedral convex set; see \cite[Theorem~19.3]{roc70a}.
\end{proof}

In a general nonlinear model, however, the set $\C$ may fail to be closed already with finite $\Omega$ and even under the no arbitrage condition.

\begin{example}\label{ex:fs}
Consider Example~\ref{ex:numeraire} in the case $T=1$, so that
\[
\C = \{c\in\M\,|\,\exists\tilde x_0\in\tilde D_0:\ c_0+c_1\le\tilde x_0\cdot(s_1-s_0)\}.
\]
Let $\Omega=\{\omega^1,\omega^2\}$, $\tilde J=\{1,2\}$, $\tilde D_0=\{(x^1,x^2)\,|\, (x^1+1)(x^2+1)\ge 1\}$, $\tilde s_0=(1,1)$ and 
\[
\tilde s_1(\omega)=
\begin{cases}
(1,2) & \text{if $\omega=\omega^1$},\\
(1,0) & \text{if $\omega=\omega^2$}.
\end{cases}
\]
Since $\tilde s^1$ is constant, we get
\[
\C = \{c\in\M\,|\,\exists\tilde x_0^2\in\tilde D_0^2:\ c_0+c_1\le\tilde x_0^2(s_1^2-s_0^2)\},
\]
where $\tilde D^2_0$ is the projection of $\tilde D_0$ on the second component. Since $\tilde D^2_0=(-1,+\infty)$, $s_1^2(\omega^1)-s_0^2=1$ and $s_1^2(\omega^2)-s_0^2=-1$, we get
\begin{align*}
\C &= \{c\in\M\,|\,\exists\tilde x_0>-1:\ c_0+c_1(\omega^1)\le\tilde x_0^2,\ c_0+c_1(\omega^2)\le -\tilde x_0^2\}\\
&=\{c\in\M\,|\,c_0+c_1(\omega^1)+c_0+c_1(\omega^2)\le 0,\ c_0+c_1(\omega^2)<1\},
\end{align*}
which is not closed even though the no arbitrage condition $\C\cap\M_+=\{0\}$ is satisfied.
\end{example}

In order to find sufficient conditions for nonlinear models with general $\Omega$, we resort to traditional closedness criteria from convex analysis; see \cite[Section~9]{roc70a}. Given an $\alpha>0$, it is easily checked that 
\[
(\alpha\star S)_t(x,\omega):=\alpha S_t(\alpha^{-1} x,\omega)
\]
defines a convex cost process in the sense of Definition~\ref{ccp}. If $S$ is positively homogeneous, we have $\alpha\star S=S$, but in general, $\alpha\star S$ decreases as $\alpha$ increases; see \cite[Theorem~23.1]{roc70a}. The upper limit
\[
S_t^\infty(x,\omega) := \sup_{\alpha>0}\alpha\star S_t(x,\omega),
\]
known as the {\em recession function} of $S_t(\cdot,\omega)$, describes the behavior of $S_t(x,\omega)$ infinitely far from the origin; see~\cite[Section~8]{roc70a}. Analogously, if $D$ is conical, we have $\alpha D=D$, but in general, $\alpha D$ gets smaller when $\alpha$ decreases. Since $D_t(\omega)$ is closed and convex, the intersection
\[
D_t^\infty(\omega) = \bigcap_{\alpha>0}\alpha D_t(\omega),
\]
coincides with the {\em recession cone} of $D_t(\omega)$; see \cite[Corollary~8.3.2]{roc70a}. 

An $\reals^J$-valued adapted process $s=(s_t)$ is called a {\em market price process} if $s_t\in\partial S_t(0)$ almost surely for every $t=0,\ldots,T$; see \cite{pen7}. Here,
\[
\partial S_t(0,\omega):=\{v\in\reals^J\,|\,S_t(x,\omega)\ge S_t(0,\omega)+v\cdot x\ \forall x\in\reals^J\}
\]
is the {\em subdifferential} of $S_t$ at the origin. It gives the set of marginal prices associated with infinitesimal trades. If $S_t(\cdot,\omega)$ happens to be smooth at the origin, then $\partial S_t(0,\omega)=\nabla S_t(0,\omega)$.

\begin{theorem}\label{thm:cl}
The set $\C$ is closed in probability if 
\[
D^\infty_t(\omega)\cap\{x\in\reals^J\,|\,S^\infty_t(x,\omega)\le 0\}=\{0\}
\]
almost surely for every $t=0,\ldots,T$. This holds, in particular, if there exists a componentwise strictly positive market price process and if $D^\infty\subset\reals^J_+$.
\end{theorem}

\begin{proof}
Let $(c^\nu)_{\nu=1}^\infty$ be a sequence in $\C$ converging to a $c$. By passing to a subsequence if necessary, we may assume that $c^\nu\to c$ almost surely. Let $x^\nu\in\N_0$ be a superhedging portfolio process for $c^\nu$, i.e.\ 
\[
x^\nu_t\in D_t,\quad S_t(x^\nu_t-x^\nu_{t-1}) + c^\nu_t\le 0
\]
almost surely for $t=0,\ldots,T$ and $x^\nu_{-1}=x^\nu_T=0$. We will show that the sequence $(x^\nu)_{\nu=1}^\infty$ is almost surely bounded. 

Assume that $x_{t-1}^\nu$ is almost surely bounded and let $a_{t-1}\in L^0$ be such that $x_{t-1}^\nu\in a_{t-1}\uball$ almost surely for every $\nu$. Defining $\underbar c_t(\omega)=\inf c^\nu_t(\omega)$ we then get that
\begin{align*}
x^\nu_t(\omega)&\in D_t(\omega)\cap\{x\in\reals^J\,|\,S_t(x-x^\nu_{t-1}(\omega),\omega)+c^\nu_t(\omega)\le 0\}\\
&\subset D_t(\omega)\cap\left[\{x\in\reals^J\,|\,S_t(x,\omega)+c^\nu_t(\omega)\le 0\} + a_{t-1}(\omega)\uball\right]\\
&\subset D_t(\omega)\cap\left[\{x\in\reals^J\,|\,S_t(x,\omega)+\underbar c_t(\omega)\le 0\} + a_{t-1}(\omega)\uball\right].
\end{align*}
By \cite[Theorem~8.4]{roc70a}, this set is bounded exactly when its recession cone consists only of the zero vector. By Corollary~8.3.3 and Theorems~9.1 and 8.7 of \cite{roc70a}, the recession cone can be written as
\[
D^\infty_t(\omega)\cap\{x\in\reals^J\,|\,S^\infty_t(x,\omega)\le 0\},
\]
which equals $\{0\}$, by assumption. It thus follows that $x^\nu_t$ is almost surely bounded and then, by induction, the whole sequence $(x^\nu)_{\nu=1}^\infty$ has to be almost surely bounded.

By Komlos's principle of subsequences (see e.g.\ \cite[Lemma~1.69]{fs4}), there is a sequence of convex combinations
\[
\bar x^\mu = \sum_{\nu=\mu}^\infty\alpha^{\mu,\nu}x^\nu
\]
that converges almost surely to an $x$. Since $c^\nu\to c$ almost surely, we also get that 
\[
\bar c^\mu := \sum_{\nu=\mu}^\infty\alpha^{\mu,\nu}c^\nu \to c\quad P\text{-a.s.}.
\]
By convexity, of $D$ and $S$,
\[
\bar x^\mu_t\in D_t,\quad S_t(\bar x^\mu_t-\bar x^\mu_{t-1})+\bar c^\mu_t\le 0
\]
and then, by closedness of $D_t(\omega)$ and lower semicontinuity of $S_t(\cdot,\omega)$,
\[
x_t\in D_t,\quad S_t(x_t-x_{t-1})+c_t\le 0.
\]
Thus, $c\in \C$ and the first claim follows.

If $s\in\partial S(0)$ is a market price process, then $s_t(\omega)\cdot x\le S_t(x,\omega)$ for every $x\in\reals^J$ and thus $s_t(\omega)\cdot x\le S^\infty_t(s,\omega)$ for every $x\in\reals^J$. If we also have $D^\infty\subset\reals^J_+$, then
\[
D^\infty_t(\omega)\cap\{x\in\reals^J\,|\,S^\infty_t(x,\omega)\le 0\}\subset \reals^J_+\cap\{x\in\reals^J\,|\,s_t(\omega)\cdot x\le 0\},
\]
which reduces to the origin when $s$ is strictly positive.
\end{proof}

The set $D_t^\infty(\omega)$ consists of portfolios that can be scaled by arbitrarily large positive numbers without ever leaving the set $D_t(\omega)$ of feasible portfolios. By \cite[Theorem~8.6]{roc70a}, the set
\[
\{x\in\reals^J\,|\,S^\infty_t(x,\omega)\le 0\}
\]
gives the set of portfolios $x$ such that the cost $S_t(\alpha x,\omega)$ is nonincreasing as a function of $\alpha$. Since $S_t(0,\omega)=0$, we also have $S_t(x,\omega)\le 0$ for every $x$ with $S^\infty_t(x,\omega)\le 0$. 

The existence of a strictly positive market price process in Theorem~\ref{thm:cl} is a natural assumption in many situations. In double auction markets, for example, it means that ask prices of all assets are always strictly positive. The condition $D^\infty_t(\omega)\subset\reals^J_+$ means that if a portfolio $x\in\reals^J$ has one or more negative components then $\alpha x$ leaves the set $D_t(\omega)$ for large enough $\alpha>0$. This holds in particular if portfolios are not allowed to go infinitely short in any of the assets.

Example~\ref{ex:fs} shows that the no arbitrage condition does not imply the conditions of Theorem~\ref{thm:cl} (in Example~\ref{ex:fs}, $D_0^\infty(\omega)=\reals\times\reals_+^2$ and $S_t^\infty(x,\omega)=s_t(\omega)\cdot x$). On the other hand, the conditions of Theorem~\ref{thm:cl} may very well hold (and thus, $\C$ be closed) even when the no arbitrage condition is violated.

\begin{example}
Let $S_t(x,\omega)=s_t(\omega)\cdot x$ where $s=(s_t)$ is a componentwise strictly positive marginal price process. It is easy to construct examples of $s$ that allow for arbitrage in an unconstrained model. Let $\bar x\in\N_0$ be an arbitrage strategy in such a model and consider another model with constraints defined by $D_t(\omega)=\{x\in\reals^J\,|\,x\ge\bar x_t(\omega)\}$. In this model, $\bar x$ is still an arbitrage strategy but now the conditions of Theorem~\ref{thm:cl} are satisfied so $\C$ is closed.
\end{example}

\section{Conclusions}\label{sec:conclusions}

We have studied superhedging without assuming the no arbitrage condition. What was crucial in our derivation of the dual characterizations in Theorem~\ref{thm:d}, was that $-p\in\rc \C$ and $p\notin\rc \C$ and that $\C$ is closed in probability. These conditions may very well hold under arbitrage; see Sections~\ref{sec:sh} and \ref{sec:cl}. For the more concrete characterizations with Lemma~\ref{lem:pid}, the theory of normal integrands and thus the properties of $S$ and $D$ in Definitions~\ref{ccp} and \ref{pcp} were essential.

Most of the results in this paper were stated in terms of the set $\C$ of claim processes hedgeable with zero cost. This means that the results are not tied to the particular market model presented in Section~\ref{sec:mm} but apply to any model where the set $\C$ is closed in probability and has the properties in Lemma~\ref{lem:convex}. In particular, one could look for conditions that yield convexity in market models with long terms prices impacts; see e.g.\ Alfonsi, Schied and Schulz~\cite{ass}.

In reality, one rarely looks for superhedging strategies when trading in practice. Instead, one (more or less quantitatively) sets bounds on acceptable levels of ``risk'' when taking positions in the market and  when quoting prices. This takes us beyond the completely riskless superhedging formulations studied in this paper; see e.g.\ \cite[Chapter~8]{fs4}. Nevertheless, superhedging forms the basis for more general formulations of pricing and hedging problems. Closedness and duality results such as the ones derived in this paper are in a significant role also in these more general frameworks.

\section*{Appendix}

The following is well-known in convex analysis but because of its importance, we give a proof.

\begin{lemma}\label{lem:rc}
The recession cone of a convex set $C$ is a convex cone. If $C$ is algebraically closed, then $y\in\rc C$ if there exists even one $x\in C$ such that $x+\alpha y\in C$ for every $\alpha>0$.
\end{lemma}

\begin{proof}
It is clear that $\rc C$ is a cone. As for convexity, let $y_1,y_2\in\rc C$ and $\lambda\in[0,1]$. It suffices to show that $x+\alpha (\lambda y_1+(1-\lambda)y_2)\in C$ for every $x\in C$ and $\alpha>0$. Since $y_i\in\rc C$, we have $x+\alpha y_i\in C$ and then, by convexity of $C$, $\lambda(x+\alpha y_1)+(1-\lambda)(x+\alpha y_2) = x + \alpha(\lambda y_1 + (1-\lambda)y_2)\in C$. 

Let $x\in C$ and $y\ne 0$ be such that $x+\alpha y\in C$ $\forall\alpha> 0$ and let $x'\in C$ and $\alpha'>0$ be arbitrary. It suffices to show that $x'+\alpha' y\in C$. Since $x+\alpha y\in C$ for every $\alpha\ge\alpha'$, we have, by convexity of $C$, 
\[
x'+\alpha'y + \frac{\alpha'}{\alpha}(x-x') = (1-\frac{\alpha'}{\alpha})x'+\frac{\alpha'}{\alpha}(x+\alpha y) \in C\quad\forall \alpha\ge\alpha'.
\] 
Since $C$ is algebraically closed, we must have $x'+\alpha'y\in C$.
\end{proof}

\bibliographystyle{plain}
\bibliography{/home/teemu/doc/sp/sp}

\end{document}